\documentclass[12pt]{article}
\usepackage{graphicx}
\usepackage{amsmath,amsthm,amsfonts,amssymb}
\usepackage{url}

\newcommand{\R}{\mathbb{R}}
\newcommand{\N}{\mathbb{N}}
\newcommand{\ket}[1]{| #1 \rangle}

\newcommand{\comment}[1]{}
\newcommand{\hide}[1]{}
\newcommand{\EE}{\mathbb E}

\renewcommand{\phi}{\varphi}

\newcommand{\sgn}{{\rm sgn}}
\newcommand{\UAG}{{\rm UVS}}
\newcommand{\npart}{\mbox{$n$-par}\-tite }

\newtheorem{theorem}{Theorem}

\newtheorem{proposition}{Proposition}
\newtheorem{lemma}{Lemma}

\title{Simulating equatorial measurements on GHZ states with finite expected communication cost}
\date{}
\author{G. Brassard\footnote{DIRO, Universit\'e de Montr\'eal, Montr\'eal, Qu\'ebec, Canada} \and M. Kaplan$^*$ }
\begin{document}
\maketitle

\begin{abstract}
The communication cost of
simulating probability distributions obtained by measuring
quantum states is a natural way to quantify quantum non-locality.
While much is known in the case of  bipartite entanglement,
little has been done in the multipartite setting. 
In this paper, we focus on the GHZ state. Specifically,
equatorial measurements lead to correlations
similar to the ones obtained with Bell states.
We~give a protocol to simulate
these measurements on the \npart GHZ state using~ $O(n^2)$~bits of
communication on average.
\end{abstract}

\section{Introduction}

The issue of non-locality in quantum physics was raised
in 1935 by Einstein, Podolsky and Rosen~\cite{EPR35}.
Thirty years later, John Bell proved that quantum physics yields correlations
that cannot be reproduced by classical local hidden variable theories~\cite{bell64}.
This momentous discovery led to the more general question of quantifying
quantum non-locality.
Not only is this question relevant for 
the foundations of quantum physics, but
it is directly related to our understanding of the computational
power of quantum resources.

A natural quantitative approach to non-locality is
to study the amount of resources required
to reproduce probabilities obtained by measuring quantum states.
In this paper,
we consider the simulation of these distributions
using classical communication.
This approach was introduced independently by several
authors~\cite{maudlin92,bct99,steiner00}.
It~led to a series of results, culminating with the protocol of Toner and Bacon
to simulate von Neumann measurements on Bell states with a single bit
of communication~\cite{tb03}.
Later, Regev and Toner extended this result by giving a
simulation of binary von Neumann measurements on arbitrary
bipartite states using two classical bits~\cite{rt07}.

We focus here on \emph{multipartite} entanglement, and more specifically on GHZ
states~\cite{GHZ89}.
Unlike the bipartite case, which has been the topic of intensive investigation,
the simulation of multipartite entanglement 
is still teeming with major open problems.

The easiest situation arises in the case of \emph{equatorial} measurements
on a GHZ state because all the marginal probability distributions obtained by
tracing out one or more of the parties are uniform. Hence, it suffices in
this case to simulate the \npart correlation,
henceforth called the \emph{full correlation}. 
(Once this has been achieved, all the marginals can easily be made uniform~\cite{gp10}.)
Making the best of this observation, Bancal, Branciard and Gisin
have given a protocol to simulate equatorial measurements on the
tripartite and fourpartite GHZ states at an \emph{expected} cost of
10 and 20~bits of communication~\cite{BBG10}.
However, the amount of communication entailed by their protocol
is unbounded in the worst case.
More recently, Branciard and Gisin impoved this in the tripartite case with
a protocol using 3 bits of communication in the worst case~\cite{BG11}.

In this paper, we give a protocol to simulate equatorial measurements on
the $n$-partite GHZ state. For any 
measurements, our protocol has an expected cost of $O(n^2)$ bits
of communication, where the expectation is taken over the inner randomness
of the protocol.

The paper is organized as follows. In the Section~\ref{sec2},
we give the structure of the distribution arising from equatorial
measurements on GHZ states, and then show how to
simulate it.
The main technical tool that we use is a protocol to
sample uniform
vectors on connected subsets of the circle. We call this
task Uniform Vector Sampling.
The protocol to sample those vectors is given in Section~\ref{sec3}.

The structure of our simulation is inspired by the protocol proposed by
Toner and Bacon to simulate von Neumann measurements on Bell states.
It can be divided in two parts. In the first part, the players communicate to sample
shared random vectors on the sphere. The distribution of
these vectors depend on the player's input.
In the second part, they apply a post-processing to compute their outputs.
While the sampling part in the protocol of Toner and Bacon is simple,
it is more involved in our case. The goal of Uniform Vector Sampling
is to sample vectors with the correct distribution.

\section{Simulating equatorial measurements}
\label{sec2}
We consider the family of GHZ states $\ket{\Psi_n} = \frac{1}{\sqrt 2} (\ket{0^n} + \ket{1^n})$,
and the distribution generated by the following process:
$n$ players each receive a qubit of $\ket{\Psi_n}$.
Each player applies a bipartite measurement to its share. Let
$\mbox{$o_i \in \{-1, 1 \}$}$ denote the output of the $i^{th}$ player.
The problem is to simulate the probability distribution over the player's
output using hidden variables and communication.

The measurement operators corresponding to equatorial measurements
are on the equator of the Bloch sphere and therefore, 
can be parametrized by a single polar angle.
Denote $\alpha_i$ the
angle corresponding to the~$i^{th}$ player's measurement.
It is known that the distribution arising from such measurements is fully characterized by
the full correlation (see~\textit{e.g.}~\cite{BBG10}).
\begin{proposition}
The distributions of the outputs $\{o_i\}$ is characterized by the following relations:
\begin{itemize}
\item The full correlation is given by $\EE \left [\prod_{i=1}^n o_i \right ]= \cos(\sum_{i=1}^n \alpha_i)$.
\item The marginal distributions are given by $\mathbb E \left [\prod_{i \in S} o_i \right ] = 0$ for all $S \varsubsetneq [n]$.
\end{itemize}
\end{proposition}
We denote $S^2$ the sphere in dimension 3 and $S^1$
the sphere in dimension 2.
A vector on $S^1$ is parametrized by a single polar angle,
or equivalently a real number modulo $2\pi$.
An interval on $S^1$ is 
a connected subset of $S^1$ or equivalently of $\R/ 2\pi \mathbb Z$.

Our simulation is based on a procedure to sample uniform vectors on intervals of $S^1$,
when the description of this subset is shared among several players.
For~$k,n \geq 1$, we introduce the following task, called Uniform Vector Sampling and denoted
$\UAG(n,k)$.
The $n$ players each receive the angles $\alpha_1, \ldots, \alpha_n$, respectively. Each player
computes a message depending on his input and on a 
public random variable $r$ and sends it to a referee.
At the end, the referee has to output a uniform angle $\theta$ on the
interval $\left[\sum_{i=1}^n \alpha_i - \pi/ 2^k, \sum_{i=1}^n \alpha_i + \pi/ 2^k\right]$.
We measure the communication cost of a protocol for $\UAG$ by
considering the total length of all messages sent from the players to the referee.

\begin{theorem}
\label{thm:sampling}
For any $n,k \geq 1$, there exists a protocol for $\UAG(n,k)$
with expected communication cost at most $n(n+k)$.
\end{theorem}

We now show how to simulate
equatorial measurement on GHZ states, given a protocol for $\UAG$.
Toner and Bacon proposed a simulation of binary measurements on Bell states,
using a single bit of communication.
In the bipartite, the correlation between the player's output
is a scalar product of two vectors on $S^2$.
We sketch their protocol.
For a vector $a \in S^2$,
denote $S^+(a)$ the half sphere centered on $a$,
$S^+(a) = \{\lambda \in S^2: \langle a, \lambda \rangle \}$.
We denote the sign function with range $\{-1, +1\}$ by $\sgn$.
Toner and Bacon prove the following theorem.
\begin{theorem}[\cite{tb03}]
\label{thm:TB}
Let $a, b$ be vectors in $S^2$, and $\lambda_1, \lambda_2$ be uniformly distributed
on $S^+(a)$. Then $$\EE [\sgn \langle \lambda_1 +\lambda_2 , b\rangle] = \langle a,b\rangle.$$
\end{theorem}

To complete the simulation, it suffices to notice that shared uniform vectors on $S^+(a)$ can be sampled
efficiently by players using shared randomness and communication, even if only
one player has a full description of $a$.
The idea is to first sample a uniform random vector on the sphere, and then the player
that knows~$a$ tells the other if he has to flip the random vector 
in order to get a vector in $S^+(a)$.
This requires to send a single bit of communication.

Our simulation is based on the following observation.
Consider
 $d = (\cos \sum_{i=1}^{n-1} \alpha_i, \sin \sum_{i=1}^{n-1} \alpha_i, 0)$ and
$a_n = (\cos \alpha_n, -\sin \alpha_n, 0)$.
These are unit vectors on
$S^1$, embeded in $\R^3$ to apply Theorem~\ref{thm:TB}.
For these vectors, we have $\langle d, a_n \rangle = \cos \sum_{i=1}^n a_i$.
Therefore, if $\lambda_1$ and $\lambda_2$ are two vectors sampled uniformly
on $S^+(d)$, Theorem~\ref{thm:TB} gives
$\EE [\sgn \langle \lambda_1+\lambda_2, a_n \rangle]=  \cos \sum_i \alpha_i$.

We now describe the simulation in more details.
The players are denoted $A_1, \ldots, A_n$. Before receiving their inputs,
they prepare a shared variable $r$, used for $\UAG$.
In addition, they prepare some shared uniform random bits $b_i \in \{-1, +1\}$ for $i=1, \ldots, n-1$.
In our simulation, we only need to apply Uniform Vector Sampling with $k=1$.

\begin{enumerate}
\item\label{step1}
For $i=1, \ldots, n-1$, the players $A_i$ run the protocol for $\UAG(n-1,1)$, sending
their messages to $A_n$.
\item \label{step2}
Using the messages he received,
$A_n$ sets $\theta_1$ uniform on the interval $\left[\sum_{i=1}^{n-1} \alpha_i - \pi/2, \sum_{i=1}^{n-1} \alpha_i + \pi/2 \right]$.
\item\label{step3}
The players repeat steps~\ref{step1} and~\ref{step2} to allow $A_n$ to sample
another angle $\theta_2$ with the same distribution.
\item\label{step4}
Player $A_n$ samples $u_1$ and $u_2$ uniformly on $[-1,1]$ and 
for $i=1,2$, sets
\begin{eqnarray*}
\phi_i &=&\arccos u_i\\
\lambda_i&=& (\cos \theta_i \cos \phi_i, \sin \theta_i \cos \phi_1, \sin \phi_i).
\end{eqnarray*}
\item\label{step5} For $i=1,\ldots, n-1$, the player $A_i$ outputs $o_i= b_i$.
\item\label{step6} The player $A_n$ outputs $o_n = \left(\prod_{i=1}^{n-1} b_i \right)\cdot
\sgn \langle \lambda_1 + \lambda_2, a_n\rangle$, where we defined 
$\mbox{$a_n= (\cos \alpha_n,- \sin \alpha_n, 0)$}$.
\end{enumerate}

After step~\ref{step3}, player $A_n$ has the complete description of two angles $\theta_1$
and~$\theta_2$ uniformly distributed on
$\left[\sum_{i=1}^{n-1}\alpha_i  - \frac \pi {2}, \sum_{i=1}^{n-1} \alpha_i  + \frac \pi {2} \right]$.
The purpose of 
step~\ref{step4} is to transform the angles in uniform random vectors on $S^+(d)$, where
$d$ is the vector with coordinates
$ (\cos \sum_{i=1}^{n-1}\alpha_i, \sin \sum_{i=1}^{n-1}\alpha_i, 0)$.
Since $d$ is on the equator, it is sufficient to assign a random
latitude to the vectors whose longitudes are $\theta_1$ and $\theta_2$.
Finally, after steps~\ref{step5} and~\ref{step6}, we have
\begin{eqnarray*}
\EE \prod_{i \in S} o_i &=& 0 \text{ for any } S \varsubsetneq [n],\\
\EE \prod_{i=1}^n o_i &=& \EE [\sgn \langle \lambda_1 + \lambda_2, a_n\rangle] \\
&  =& \cos \sum_{i=1}^n \alpha_i \text{  by Theorem~\ref{thm:TB}.}
\end{eqnarray*}
Sampling the angles $\theta_1$ and $\theta_2$,
 can be done with $O(n^2)$ expected bits of communication. Therefore, the whole
 protocol can be done with $O(n^2)$ expected bits of communication.

\section{Uniform Vector Sampling}
\label{sec3}

The goal of this section is to prove Theorem~\ref{thm:sampling}.
Observe that in the simulation, we only need the case $k=1$. Nevertheless,
our inductive proof requires to prove the stronger statement given in Section~\ref{sec2}.

\subsection*{The base case: $n=1$}
For $n=1$, there is a single input $\alpha_1$. Fix $k \geq 1$.
Let $\delta$ be chose uniformly at random on $S^1$.
The player sends
$$t = \min \left \{i \in \N: \delta + i \frac \pi {2^{k-1}} \in [\alpha_1 - \frac \pi {2^k}, \alpha_1 + \frac \pi {2^k}] 
\right \}$$
to the referee, who computes
$\theta = \delta + t \frac{\pi}{2^{k-1}}$.
The resulting angle $\theta$ is uniformly distributed on
$[\alpha_1 - \frac{\pi}{2^k}, \alpha_1 + \frac{\pi}{2^k}]$.
Notice that since $t \in [2^{k} -1]$, the length of the message is at most $k$.

\subsection*{The induction step}
Let $n>1$. The following Lemma is the main technical tool
that we use for the induction. It explains how to generate
uniformly distributed variables from specific non-uniform ones.
We first give and prove Lemma~\ref{lm:tech}, and then use it
to prove the induction.
\begin{lemma}
\label{lm:tech}
Let $\mathcal D_i^-$ denote the uniform distribution on $[0, 1/2^i]$ and
$\mathcal D_i^+$ denote the uniform distribution on $[1-1/2^i, 1]$. Let $\mathcal D$
be the distribution on $t$ defined by the following procedure:
\begin{itemize}
\item Pick an integer $i \geq 0$ with probability $1/2^{i+1}$, and $r$ uniform in $\{-1, +1\}$.
\item If $r=-1$, sample $t_1, t_2 \sim \mathcal D_i^-$.
\item Otherwise, sample $t_1, t_2 \sim \mathcal D_i^+$.
\item Set $t=t_1+t_2$.
\end{itemize}
Then $\mathcal D$ is the uniform distribution on $[0,1]$.
\end{lemma}
\begin{proof}
Denote $U_i^-=[0, 1/2^i]$ and $U_i^+ =[1-1/2^i, 1]$. We define the density functions associated
to the distributions $\mathcal D_i^+$ and $\mathcal D_i^-$,

\begin{minipage}[b]{0.4\linewidth}
\raggedleft
$$
f_i^+(x) = \begin{cases} 2^i & \text{ if } x\in U_i^+,\\ 0 & \text{ otherwise,} \end{cases}$$
\end{minipage}
\begin{minipage}[b]{0.5\linewidth}
$$\text{and }f_i^-(x) = \begin{cases} 2^i & \text{ if } x\in U_i^-,\\ 0 & \text{ otherwise.} \end{cases}$$
\end{minipage}\\
By definition, the density $\rho_i$ of $t_1+t_2$ for a fixed $i$ is
$$\rho_i = \frac 1 2( f_i^+ * f_i^- + f_i^- * f_i^-),$$
where $*$ denotes the convolution product of two functions.
By direct calculation, we have
$$(f_i^- * f_i^-)(x) =
\begin{cases} 2^{2(i+1)}x &\text{ if } x \in [0, 1/2^{i+1}],\\
		2^{i+2}-2^{2(i+1)}x & \text{ if } x\in [1/2^{i+1}, 1/2^i],\\
		0\text{ otherwise.}\end{cases}$$
and
$$(f_i^+ * f_i^+)(x)= (f_i^- * f_i^-)(1-x).$$
Let $\rho$ denote the density of the distribution $\mathcal D$. We now calculate $\rho(x)$.
Notice that $f_0^- = f_0^+$, and for $i>0$, $f_i^-$ and $f_i^+$ have disjoint supports.
Assume that $x< 1/2$ (the other case is similar). 
In that case, $f_i^+(x)= 0$ for any $i >0$.
Let $j = \max \{j': x \in [0, 1/2^{j'}]\}$ and notice that $f_i^-(x)=0$ for any $i > j$. We have
\begin{eqnarray*}
\rho(x) &=& \sum_{i=0}^\infty \frac 1 {2^{i+1}} \rho_i(x)\\
&=&\frac 1 2 \cdot  \frac 1 2 (f_0^+ * f_0^+)(x) + \frac 1 2 \sum_{i=0}^j  \frac 1 {2^{i+1}} (f_i^- * f_i^-)\\
&=& x + \frac 1 2 \left[ \sum_{i=0}^{j-1} \frac 1 {2^{i+1}} 2^{2(i+1)} x + \frac 1 {2^{j+1}}(2^{j+2} - 2^{2(j+1)}x)\right]\\
&=& x+  \left(\sum_{i=0}^{j-1} 2^{i}\right)x +1 - 2^j x\\
&=& x+ (2^j-1)x+1 - 2^j x\\
& =&1\\
\end{eqnarray*}
which concludes the proof.
\end{proof}

\begin{figure}[t]
\begin{minipage}[b]{0.5\linewidth}
\includegraphics[width=0.8\linewidth]{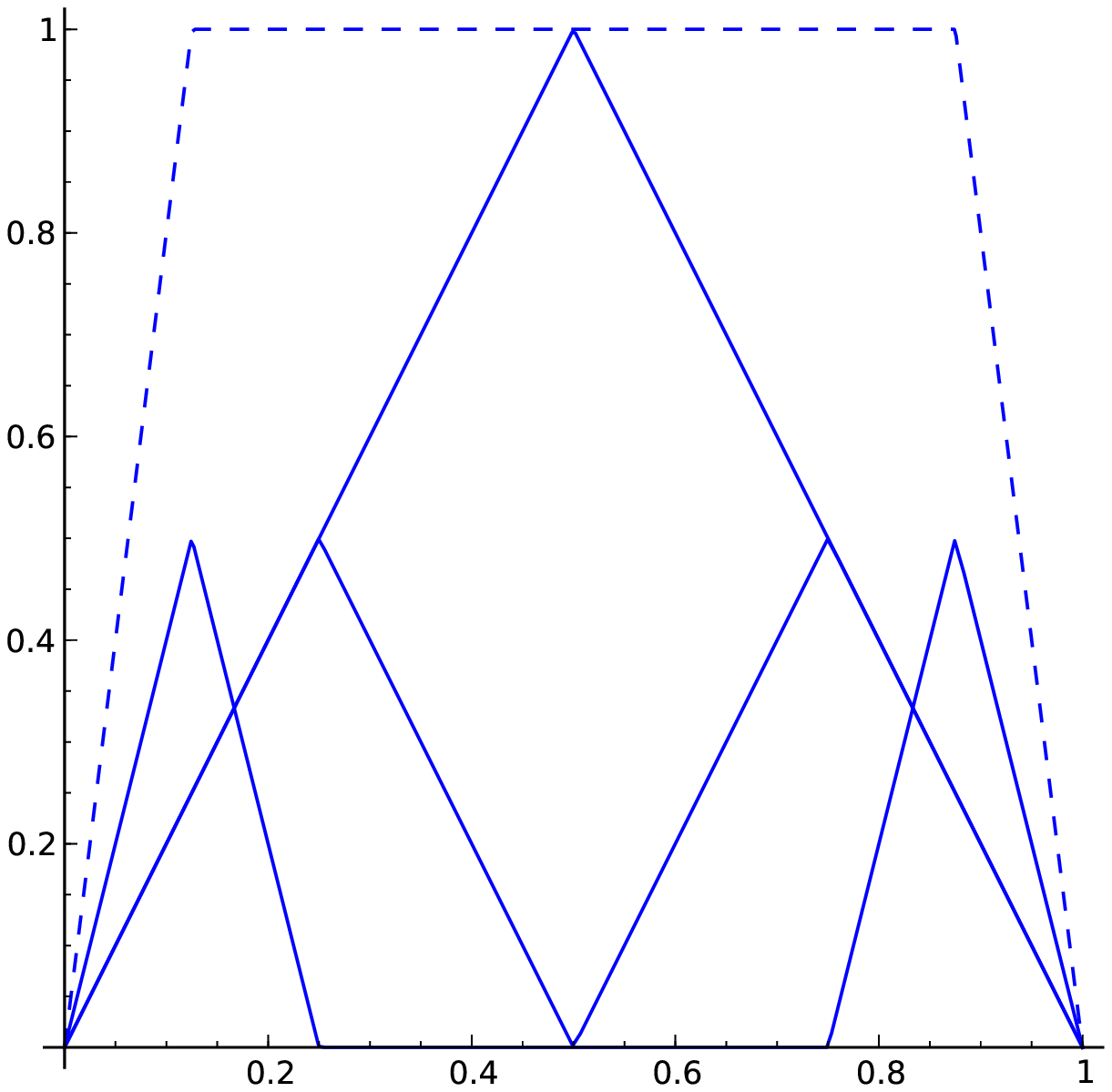}
\end{minipage}
\begin{minipage}[b]{0.5\linewidth}
\includegraphics[width=0.8\linewidth]{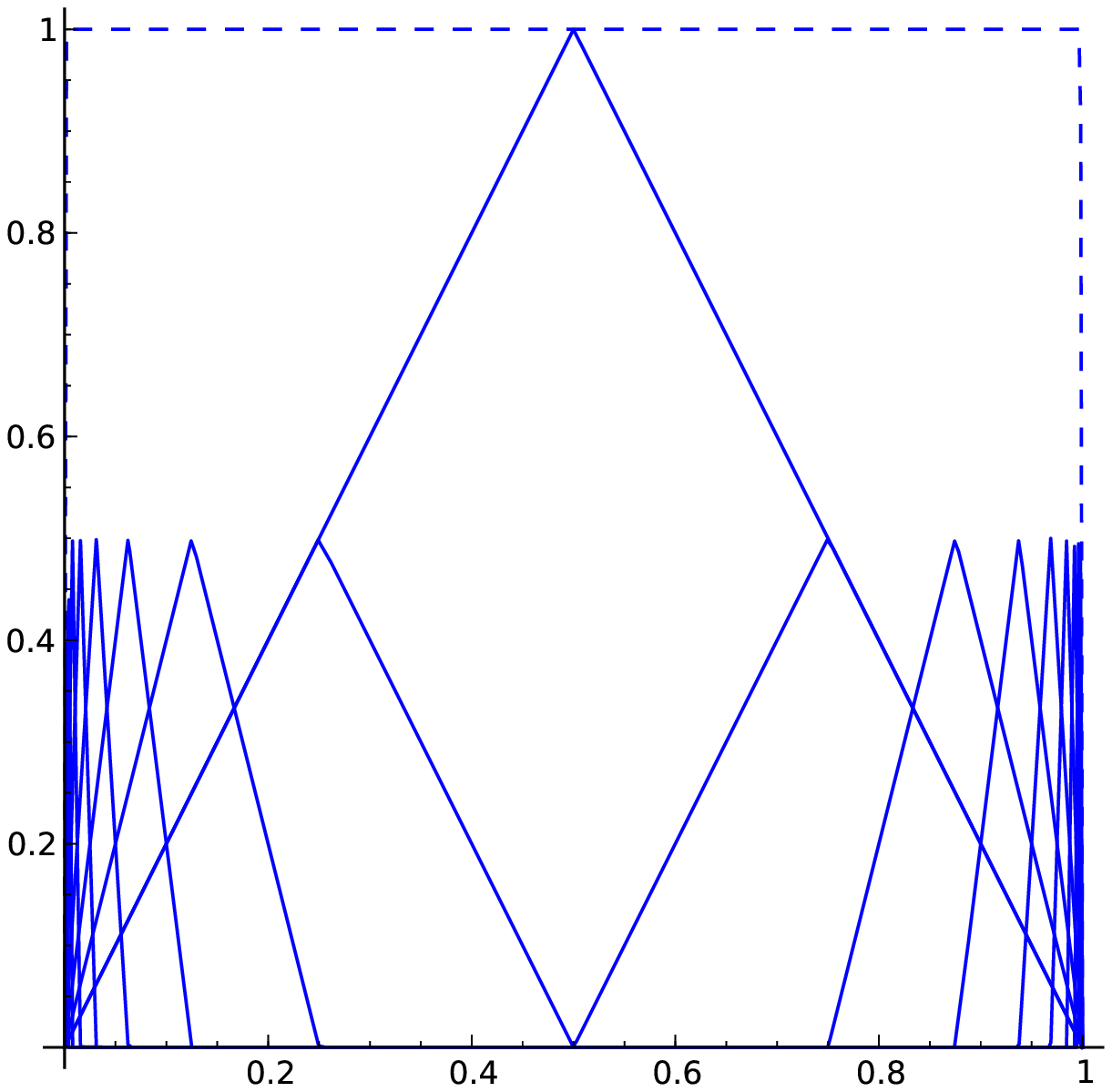}
\end{minipage}
\caption{The density functions $f_i^+$ and $f_i^-$, for $i\leq 3$ and $i\leq8$. Each
density function is scaled down by the probability of sampling it in Lemma~\ref{lm:tech}.
$f_0^+$ and $f_0^-$ are equal.
The dashed curves represent the sum of the represented density functions.}
\end{figure}

The induction hypothesis is that for any $k\geq 1$,
it is possible for $n-1$ players to each send a message to another party
such that he outputs an angle $\theta$ uniformly distributed
on $[\sum_{i=1}^{n-1} \alpha_i - \pi/2^{k} , \sum_{i=1}^{n-1} \alpha_i + \pi/2^{k}]$.

Before receiving their inputs, the players prepare the following random elements:
\begin{itemize}
\item an integer $j\geq 0$ chosen with probability $p(j)=1/2^{j+1}$,
\item $b$ uniform in $\{-1, +1\}$,
\item the random elements required to run
$\UAG(n-1,k+j+1)$,
\item the random elements required to run $\UAG(1, k+j+1)$.
\end{itemize}

The protocol proceeds as follows. The $n-1$ players first players
send to the referee the messages corresponding to
$\UAG(n-1, k+j+1)$. The referee uses them to prepare
$\theta_1$ uniform on $[\sum_{i=1}^{n-1} \alpha_i -\pi/2^{k+j+1}, 
\sum_{i=1}^{n-1} \alpha_i +\pi/2^{k+j+1}]$.
The $n$-th player sends to the referee the message
corresponding to $\UAG(1, k+j+1)$. The referee uses
it to prepare $\theta_2$ uniform on $[\alpha_n - \pi/2^{k+j+1}, 
\alpha_n + \pi/2^{k+j+1}]$.
Finally, the referee outputs 
$\theta = \theta_1+\theta_2+ b\cdot \frac{\pi}{2^{k}}\left(1 - \frac 1{2^j}\right)$.
 
We now analyze the protocol and prove that $\theta$ is uniform on the interval 
$\left[ \sum_{i=1}^n \alpha_i - \pi/2^k, \sum_{i=1}^n \alpha_i + \pi/2^k \right]$.
To apply Lemma~\ref{lm:tech}, we need to rescale the random variables $\theta_1$ and
$\theta_2$.
We split the term $b\cdot \frac \pi{2^k} ( 1- \frac 1 {2^j})$ in two parts and
think of each as a shift of $\theta_1$ and $\theta_2$
in a direction that depends on the bit $b$. Each angle is shifted in the same
direction.

Let $v_{1,j}^-$ and $v_{2,j}^-$ be uniform random variables on $[0, 1/2^j]$
and~$v_{1,j}^+$ and $v_{2,j}^+$ be uniform random variables on $[1-1/2^j, 1]$.
Let $\mathbf T_{1,j}$ denote the random variable
$\theta_{1|j}$, that is, the random variable generated by $\UAG(n-1, k+j+1)$ for
a fixed value of $j$.
The shifted random variable $ \mathbf T_{1,j}+ b \frac \pi {2^{k+1}}
\left(1 - \frac 1{2^j}\right) $ is uniform
\begin{itemize}
\item either on $\left [\sum_{i=1}^{n-1} \alpha_i - \pi/2^{k+1}, \sum_{i=1}^{n-1} \alpha_i - \pi/{2^{k+1}}+ \pi/2^{k+j}\right ]$ if $b=-1$,
\item or on $\left [\sum_{i=1}^{n-1} \alpha_i + \pi/2^{k+1} - \pi/2^{k+j}, \sum_{i=1}^{n-1} \alpha_i + \pi/{2^{k+1}}\right ]$ if $b=+1$.
\end{itemize}
Using, $v_{1,j}^+$ and $v_{1,j}^-$, we can rewrite
$$
\mathbf T_{1,j}+ b \frac \pi {2^{k+1}} \left(1 - \frac 1{2^j}\right) =
\begin{cases} \sum_{i=1}^{n-1} \alpha_i - \pi/2^{k+1} + v_{1,j}^- \cdot \pi/{2^k} &\text{ if $ b=-1$}\\
\sum_{i=1}^{n-1} \alpha_i - \pi/2^{k+1} + v_{1,j}^+ \cdot \pi/{2^k} &\text{ if $ b=+1$}\\
\end{cases}
$$
Similarly, let $\mathbf T_{2,j}$ denote the random variable
$\theta_{2|j}$.
Using $v_{2,j}^-$ and $v_{2,j}^+$, it can be written
$$
\mathbf T_{2,j}+ b \frac \pi {2^{k+1}} \left(1 - \frac 1{2^j}\right) =
\begin{cases} \alpha_n - \pi/2^{k+1} + v_{2,j}^- \cdot \pi/{2^k} &\text{ if $ b=-1$}\\
\alpha_n - \pi/2^{k+1} + v_{2,j}^+ \cdot \pi/{2^k} &\text{ if $ b=+1$}\\
\end{cases}
$$
For the sum, we get the expression
$$\theta = \mathbf T_{1,j} + \mathbf T_{2,j} + b\cdot \frac{\pi}{2^{k}}\left(1 - \frac 1{2^j}\right)
= \sum_{i=1}^n \alpha_i -\pi/2^k + v_{j,b} \cdot \pi/2^{k-1} $$
where $v_{j,b}$ is the sum of $v_{1,j}^+$ and $v_{2,j}^+$ if $b=+1$ and
the sum of $v_{1,j}^-$ and $v_{2,j}^-$ if $b=-1$.
According to Lemma~\ref{lm:tech}, when taking the expectation over $j$ and $b$,
$v_{j,b}$ is uniform on $[0,1$]. In consequence, $\theta$ is uniform
on the interval $\left [\sum_{i=1}^n \alpha_i -\pi/2^k, \sum_{i=1}^n \alpha_i +\pi/2^k\right]$.

It remains to bound the expected length of messages.
Denote $l_{n,k}$ the expected sum of the messages length.
We already know that $l_{1,k} \leq k$  for any~$k$.
Fix $n>1$. Analyzing our protocol, we get the induction:
\begin{eqnarray*}
l_{n,k} &=& \sum_{j\geq 0} \frac 1 {2^{j+1}} (l_{n-1, k+j+1} + l_{1,k+j+1}),\\
&\leq & \sum_{j\geq 0} \frac 1 {2^{j+1}} l_{n-1, k+j+1} + \sum_{j\geq 0} \frac{k+j+1}{2^{j+1}},\\
&\leq & \sum_{j \geq 0} \frac 1 {2^{j+1}} l_{n-1, k+j+1} + k.
\end{eqnarray*}
The induction hypothesis is that $l_{n-1,k+j+1} \leq (n-1)(n+k+j)$. We plus this expression and
get
\begin{eqnarray*}
l_{n+1,k} &\leq& \sum_{j\geq 0} \left(\frac  {(n-1)(nk+j)}{2^{j+1}} \right ) + k,\\
&\leq& (n-1)(n+k+1) +k,\\
&\leq& n(n+k),
\end{eqnarray*}
which concludes the proof.

\section{Conclusion}
We gave a protocol to simulate equatorial measurements on
the \npart GHZ state, using $O(n^2)$ bits on average.
Our protocol is in two parts. Firstly, we reduce the problem
to sampling vectors on regions of the $S^1$. Secondly,
we give a procedure to sample the vectors, called
Uniform Vector Sampling. This scheme is
inspired by the protocol of Toner and Bacon to simulate
von Neumann measurements on Bell States.

Our work leads to an obvious question. Is it possible
to transform our protocol into a protocol that is bounded
in the worst case? To solve this question, it enough to
give a protocol for $\UAG$ that use bounded communication
in the worst case.
Uniform Vector Sampling could also be considered as a task of independent
interest or be applied in other contexts.

Our work, like others on the same topic, considers only
equatorial measurements. The simulation of more general measurements
is an intriguingly hard question. The main difference is that they
lead to non-uniform marginals.
In the bipartite case, an analogous
problem arises when considering non-maximally entangled states.
It may seem that modifying local marginals is easy once the
correlation is simulated. Unfortunately, local transforms usually
also modify the full correlation.

\section*{Aknowledgement}
We thank Nicolas Gisin and Cyril Branciard  for interesting discussions.

\end{document}